\documentclass[thmrestate,runningheads,anonymous]{llncs}
\usepackage{amsmath,amssymb,geometry}
\usepackage{xspace}
\usepackage{framed} 
\usepackage{thmtools}
\newcommand{\rETH}{r\textsf{ETH}\xspace}
\newcommand{\ALP}{{\sc $(A, \ell)$-Path Packing}\xspace}
\usepackage{thmtools}

\title{A Deterministic Separation Lemma}
\author{
Abhishek Sahu\\
\texttt{abhishek.sahu@niser.ac.in}\\[1ex]
National Institute of Science Education and Research (NISER), India\\
and\\
Homi Bhabha National Institute (HBNI),\\
Training School Complex, Anushakti Nagar, Mumbai 400094, India
}

\begin{document}
\author{Abhishek Sahu\inst{1,2}}
\authorrunning{A. Sahu}

\institute{
National Institute of Science Education and Research (NISER), Bhubaneswar, India
\and
Homi Bhabha National Institute (HBNI), Training School Complex, Anushakti Nagar, Mumbai 400094, India
\email{abhisheksahu@niser.ac.in}
}

\maketitle

\begin{abstract}
The \emph{Separation Lemma} is a simple yet powerful tool, akin to the well-known \emph{Isolation Lemma}, that guarantees the uniqueness of certain set sums. Bandopadhyay et al.\ introduced this lemma to establish lower bounds for the \ALP problem with respect to certain structural parameters, relying on random weight assignments in the process. The lemma's applicability extends well beyond that specific work, especially in proving hardness results. However, while effective, these hardness results inherently rely on probabilistic assumptions. In this work, we give a fully \emph{deterministic} construction for the weight assignment required by the Separation Lemma. We provide formal proofs of correctness, explicit examples, and show how deterministic weights can replace randomized ones, thereby derandomizing existing hardness results for path-packing problems. Our exposition highlights a clear progression from the original randomized foundations to deterministic constructions and their practical implications.
\end{abstract}

\section{Introduction}

Randomized constructions, such as the \emph{Isolation Lemma}~\cite{MulmuleyVaziraniVazirani1987}, have long been recognized as powerful tools in combinatorial optimization.  
The lemma states that if each element of a ground set is assigned an integer weight chosen independently and uniformly from a bounded range, then with high probability there exists a \emph{unique} minimum-weight set in any given family of subsets.  
This guarantee of a uniquely optimal solution under random, polynomially bounded weights is a cornerstone of many algorithms and hardness reductions: it enables parallel algorithms for perfect matching, simplifies counting arguments, and underlies reductions that require isolating a single combinatorial object among exponentially many candidates.
\begin{proposition}[Isolation Lemma, Mulmuley--Vazirani--Vazirani~\cite{MulmuleyVaziraniVazirani1987}]\label{prop:random-isolation}
Let $(X, \mathcal{F})$ be a finite set system, where $\mathcal{F} \subseteq 2^{X}$ is a family of subsets of $X$.  
Assign to each element $x \in X$ an integer weight $w(x)$ chosen independently and uniformly at random from the set $[M] = \{1,2,\dots,M\}$.  
Then, with probability at least 
\[
1 - \frac{|X|}{M},
\] 
there exists a \emph{unique} set $S^* \in \mathcal{F}$ whose total weight 
\[
w(S) = \sum_{x \in S} w(x)
\]
is minimum among all sets in $\mathcal{F}$. In particular, choosing $M = 2|X|$ guarantees that with probability at least $1/2$ the family $\mathcal{F}$ has a unique minimum–weight set.
\end{proposition}

 The \emph{Separation Lemma}, introduced by Bandopadhyay et al.~\cite{DBLP:conf/mfcs/BandopadhyayBMS24}, can be seen as a conceptual sister to the \emph{Isolation Lemma}: though simpler in its formulation, it offers a similarly broad potential for applications, particularly in establishing lower bounds for structured combinatorial problems.

\begin{proposition}[Separation Lemma, Bandopadhyay et al.~\cite{DBLP:conf/mfcs/BandopadhyayBMS24}]\label{prop:random-separation}
Let $(X, \mathcal{F})$ be a finite set system, where $\mathcal{F} \subseteq 2^X$ is a family of subsets of $X$. Suppose each element $x \in X$ is assigned a weight $w(x)$ independently and uniformly at random from the set $[M] = \{1, 2, \dots, M\}$. Then, with probability at least 
\[
1 - \frac{{|\mathcal{F}| \choose 2}}{M},
\] 
the induced weights of all sets in $\mathcal{F}$ are distinct; that is, for all distinct $S_1, S_2 \in \mathcal{F}$, 
\[
w(S_1) = \sum_{x \in S_1} w(x) \neq \sum_{x \in S_2} w(x) = w(S_2).
\]

In particular, choosing $M=2|\mathcal{F}|^2$ ensures that with probability at least $1/2$ every set in $\mathcal{F}$ has a unique weight.
\end{proposition}

 Unlike the well-known \emph{Isolation Lemma}, which remains inherently randomized with no fully deterministic construction known despite decades of research, the \emph{Separation Lemma} admits a deterministic counterpart as we prove in this work (albeit with a polynomial explosion in range of weight values). Note that the central insight underlying the \emph{Separation Lemma} is that, when the number of objects or sets to be separated is polynomial in the input size, one can assign random, polynomial bounded integer weights so that all weights of sets/objects are distinct with a \emph{sufficiently high probability}. Our \emph{Deterministic Separation Lemma} builds on this idea, providing explicit weight assignments that replace randomized weight assignments making the lemma completely deterministic. The versatility of this approach is reflected in a range of illustrative applications:

\begin{itemize}
    \item \textbf{Paths in Trees:} By assigning weights to edges, every simple path can be guaranteed to have a unique total weight facilitating deterministic reductions to problems like {\sc Exact-length Path Packing} etc.

    \item \textbf{Triplets of Vertices (3D Matching):} Assigning weights to vertices ensures that weights of all triplets are distinct, facilitating deterministic reductions to problems like {\sc Bin filling} etc.

    \item \textbf{Triangles (3-Sized Sets of Edges):} Weighting edges so that every triangle weight is unique allows reductions in triangle-constrained combinatorial problems.
\end{itemize}

The common thread across these examples is that polynomial many sets or objects can be separated with polynomial bounded weights, providing deterministic guarantees.
Despite its usefulness, the \emph{Separation Lemma} in its original form is inherently probabilistic, limiting its applicability in contexts that require deterministic constructions. In this work, we focus on overcoming this limitation by developing a fully deterministic approach for assigning weights to elements of a set $X = \{x_1, \dots, x_n\}$, ensuring that all set weights $w(S) = \sum_{x \in S} w(x)$ are distinct for $S \subseteq X$. In particular, 

\begin{framed}
\begin{theorem}[Deterministic Separation Lemma]\label{lem:deterministic-separation}
Let $(X, \mathcal{F})$ be a finite set system, where $\mathcal{F} \subseteq 2^X$ is a family of subsets of $X$. Then there exists a deterministic assignment of positive integer weights 
\[
w: X \to \{1, 2, \dots, \gamma^2|\mathcal{F}|^2\}
\] 
such that for all distinct $S_1, S_2 \in \mathcal{F}$, \[w(S_1) = \sum_{x \in S_1} w(x) \neq \sum_{x \in S_2} w(x) = w(S_2),\] where $\gamma$ is the cardinality of a largest set in $\mathcal{F}$ and can always be upper-bounded by $n$.
\end{theorem}
\end{framed}

\newcounter{thmsave}
\setcounter{thmsave}{\value{theorem}}

\makeatletter
\renewcommand{\thetheorem}{\ref{lem:deterministic-separation}}
\makeatother

This deterministic construction not only \emph{almost} replicates the guarantees of the randomized \emph{Separation Lemma} but also strengthens existing hardness results by removing probabilistic assumptions, offering a method with broad potential applications.


\section{Deterministic Separation Lemma (DSL)}




\paragraph{Intuition.} 
We assign weights using a greedy strategy by processing the elements of $X$ one by one in a fixed order. 
At each step $i$, we maintain the following invariant: \emph{all sets in $\mathcal{F}_{\leq i}=\bigcup_{j=1}^{i}\mathcal{F}_j$ where $\mathcal{F}_j=\{S\cap\{x_1,\cdots,x_j\}| S\in \mathcal{F}\}$, have distinct total weights}. Note that $\mathcal{F}_j$ is basically projection of all sets in $\mathcal{F}$ onto first $j$ elements. We further define a set to achieve this invariant thoughout our weight assignment process, we maintain a set of forbidden values at each step, consisting of the weights of all sets in $\mathcal{F}_{\leq i}$ and all pairwise differences of these weights; we denote this set by $G_i$. When considering the next element $x_{i+1}$, we assign it the \emph{smallest positive integer} not in $G_i$, which is sufficient to preserve the same invariant after adding $x_{i+1}$. We then update the forbidden set $G_{i+1}$ based on $\mathcal{F}_{\leq i+1}$ and continue this process until all elements have been assigned weights. At each step $i$, the number of forbidden values can be upper bounded by $|\mathcal{F}_{\leq i}|^2$, which in turn bounds the weight assigned to the next element. Repeating this greedy strategy for all $n$ elements yields a deterministic weight assignment in which the total weight of every set in $\mathcal{F}$ is unique, and all element weights lie in the range $[1,|\mathcal{F}|^2]$. We provide the formal arguments below.

\begin{theorem}[Deterministic Separation Lemma]
\begingroup
  \edef\@currentlabel{\ref{lem:deterministic-separation}}%
  \label{lem:deterministic-separation1}%
\endgroup
Let $(X, \mathcal{F})$ be a finite set system, where $\mathcal{F} \subseteq 2^X$ is a family of subsets of $X$. Then there exists a deterministic assignment of positive integer weights 
\[
w: X \to \{1, 2, \dots, \gamma^2|\mathcal{F}|^2\}
\] 
such that for all distinct $S_1, S_2 \in \mathcal{F}$, \[w(S_1) = \sum_{x \in S_1} w(x) \neq \sum_{x \in S_2} w(x) = w(S_2),\] where $\gamma$ is the cardinality of a largest set in $\mathcal{F}$ and can always be upper-bounded by $n$.
\end{theorem}

\medskip

\begin{proof}

We construct the weights inductively according to an arbitrary fixed ordering $x_1< x_2< \dots< x_n$ of $X$. At each step $i$, we maintain the following invariant:

\begin{itemize}
    \item For any two distinct sets $A, B \in {\mathcal{F}_{\leq i}}$, their weights are distinct, that is,
    $w(A)  \neq  w(B)$.
\end{itemize}

To maintain the invariant when assigning a weight to $x_{i+1}$, we consider all sets contained in $\mathcal{F}_{\leq i}$. We then assign $x_{i+1}$ the \emph{smallest positive integer} that is neither the weight of any set in $\mathcal{F}_{\leq i}$ nor the difference of weights of any two sets in $\mathcal{F}_{\leq i}$; we refer to this collection of forbidden values as $G_i$. In other words, at step $i+1$, $w(x_{i+1})$ is chosen as the smallest positive integer not in $G_i$, ensuring that the invariant is preserved as proved below with an inductive argument. By this construction, all sets contained in $\mathcal{F}_{\leq i+1}$ continue to have distinct weights.

\medskip

\noindent\textbf{Base case:} Assign $w(x_1) = 1$. Observe that the set $\{x_1\}$, as well as the empty set $\phi$, have distinct weights. Hence, the invariant is satisfied for $i = 1$.
\medskip

\noindent\textbf{Inductive hypothesis:} Assume that weights $w(x_1), \dots, w(x_i)$ have been assigned such that the invariant is satisfied for the first $i$ elements; that is, all sets $A \in \mathcal{F}_{\leq i}$ have distinct total weights.
 \medskip

\noindent\textbf{Inductive step:} Consider $x_{i+1}$. We define the set of forbidden values
\[
G_i = \Big\{ w(A) \;:\; A \in \mathcal{F}_{\leq i} \Big\} \;\cup\; \Big\{ |w(A) - w(B)| \;:\; A, B \in \mathcal{F}_{\leq i}, A \neq B \Big\}.
\]
By construction, the number of forbidden values at step $i$ satisfies $|G_i| = |\mathcal{F}_{\leq i}| + \binom{|\mathcal{F}_{\leq i}|}{2} < |{\mathcal{F}}_{\leq i}|^2$. This ensures that there exists at least one available integer in the range $\{1, \dots, |\mathcal{F}_{\leq i}|^2\}$ for assigning $w(x_{i+1})$.
Choose $w(x_{i+1})$ as the smallest positive integer not in $|\mathcal{F}_i|^2$.  
\medskip

\noindent\textbf{Correctness of invariant at current step:} Suppose the invariant is not true, implying the existence of two distinct sets $A, B \in \mathcal{F}_{\leq i+1}$ with $w(A) = w(B)$ after assigning $w(x_{i+1})$. We consider the following three exhaustive cases. 

\begin{itemize}
    \item \textbf{Case 1:} \textbf{(}Both $A$ and $B$ contain $x_{i+1}.$\textbf{)} Let $A' = A \setminus \{x_{i+1}\}$ and $B' = B\setminus \{x_{i+1}\}$.  Note that since $A, B \in \mathcal{F}_{\leq i+1}$, then $A',B' \in \mathcal{F}_{\leq i}$. But $w(A)=w(B)\implies w(A') = w(B')$,  which contradicts the inductive hypothesis that two distinct sets in $\mathcal{F}_{\leq i}$ have identical weights.

\medskip

    \item \textbf{Case 2:} (Exactly one of $A, B$ contains $x_{i+1}$.) Without loss of generality, let A contain $x_{i+1}$ and $A' = A \setminus \{x_{i+1}\}$ and $B \subseteq \{x_1, \dots, x_i\}$. Again note that $A',B \in \mathcal{F}_{\leq i}$. And, since $w(x_{i+1})$ is a positive integer and $w(A)=w(B)$, it must be that $w(A')<w(B)$. And, $w(A') + w(x_{i+1}) = w(B) \implies w(x_{i+1}) = w(B) - w(A')=|w(B)-w(A')| \in G_i$, a contradiction to the fact that $x_{i+1}$ is assigned a positive forbidden weight.
\medskip

    \item \textbf{Case 3:} (Neither $A$ nor $B$ contains $x_{i+1}$.) In this case, both sets are contained in $\mathcal{F}_{\leq i}$ , so by the inductive hypothesis, their weights must be distinct, a contradiction.
\end{itemize}

Hence, all sets in $\mathcal{F}_{\leq i+1}$ have distinct weights where $x_{i+1}$ is assigned a weight in $[1,|\mathcal{F}_{\leq i}|^2]$. Since $\mathcal{F}_n\subseteq \mathcal{F}_{\leq n}$ and contains all sets in $\mathcal{F}$, by induction all these sets get distinct weights. And, the total number of sets in $\mathcal{F}_{\leq n}$ can be upper-bounded by $|\mathcal{F}|\times \gamma$ as any set $S$ in $\mathcal{F}$ has at most $|S|$ many projected sets (including itself) in  $\mathcal{F}_{\leq n}$. These facts together imply the stated Theorem~\ref{lem:deterministic-separation1}.
\end{proof}

From the proof above, it is also clear that one may choose to use the bound $[1, \dots, |\mathcal{F}_{\leq n}|^2]$ instead of $[1, \dots, \gamma^2|\mathcal{F}|^2]$ whenever convenient. We express this as the following theorem, which can be applied when $|\mathcal{F}_{\leq n}|^2$ is significantly smaller than $\gamma^2 |\mathcal{F}|^2$. Our next example follows as an immediate observation of this fact.

\begin{theorem}[An Alternate Deterministic Separation Lemma]
Let $(X, \mathcal{F})$ be a finite set system, where $\mathcal{F} \subseteq 2^X$ is a family of subsets of $X$.  
For an arbitrary ordering of the elements $x_1< x_2< \dots< x_n$ of $X$, let $\mathcal{F}_{\leq n}$ denote the collection of all projections of sets in $\mathcal{F}$ onto their first $i$ elements, for $i \in[ 1, n]$. Formally,
$
\mathcal{F}_{\leq n} = \{\, S \cap \{x_1, \dots, x_i\} \mid S \in \mathcal{F},\, 1 \le i \le n \,\}.$
\medskip

\noindent Then there exists a deterministic assignment of positive integer weights 
\[
w: X \to \{1, 2, \dots, |\mathcal{F}_{\leq n}|^2\}
\]
such that for all distinct $S_1, S_2 \in \mathcal{F}$,
\[
w(S_1) = \sum_{x \in S_1} w(x) \neq \sum_{x \in S_2} w(x) = w(S_2),
\]
\end{theorem}

\medskip

\section{Deterministic Separation Lemma on Contiguous Sets /Intervals}

\begin{definition}[Contiguous Sets/Intervals]
Let $X = \{x_1, x_2, \dots, x_n\}$ be a totally ordered set. A subset $I \subseteq X$ is called \emph{contiguous} if there exist indices $1 \le i \le j \le n$ such that $I = \{x_i, x_{i+1}, \dots, x_j\}$. Such a set $I$ is also referred to as an \emph{interval} of $X$.
\end{definition}

\begin{proposition}[Deterministic Separation Lemma for Intervals]\label{prop:interval-weight} Let $\mathcal{I} \subseteq 2^X$ denote the family of all contiguous sets/intervals of $X= \{x_1, x_2, \dots, x_n\}$. There exists a deterministic assignment of positive integer weights $w: X \to \{1, 2, \dots, |\mathcal{I}|^2(\leq n^4)\}$ such that for all distinct intervals $I_1, I_2 \in \mathcal{I}$, $w(I_1) \neq  w(I_2)$.
\end{proposition}

\begin{proof}

Since $\mathcal{I}=\{[x_i,x_j]:1\le i\le j\le n\}$, we have $|\mathcal{I}|=n(n+1)/2$. Any projection of an interval onto the first $i$ elements is again an interval, so $|\mathcal{I}_{\le i}|\le |\mathcal{I}|\le n^{2}$. Consequently, in the proof of the \emph{Deterministic Separation Lemma}, instead of selecting weights for $x_1,x_2,\dots,x_n$ from the range $[1,\gamma^{2}|\mathcal{I}|^{2}]$, one may choose them from $[1,|\mathcal{I}_{\le n}|^{2}(\leq n^4)]$. The inductive argument in the correctness proof of DSL shows that this smaller range still suffices to assign distinct total weights to all contiguous sets/intervals.

\end{proof}


\section{A Few Immediate Implications.}
\subsection{W-Hardness of \ALP\ Parameterized by Distance to a Path}

 In the context of the \ALP  problem \cite{DBLP:conf/mfcs/BandopadhyayBMS24}, the authors used (Randomized) \emph{Separation Lemma} to randomly assign weights to elements/vertices of the ground set to ensure that each pairs of disjoint intervals together had a unique total weight. This allowed them to prove the following conditional hardness result:

\begin{proposition}
Unless {\rETH} fails, there is no randomized algorithm for {\ALP} which runs in $f({\sf dtp}(G) + |A|) \cdot n^{o({\sf dtp}(G) + |A|)}$-time and is correct with probability at least $2/3$ where ${\sf dtp}(G)$ is distance to a path graph.
\end{proposition}

We introduce the family \(\mathcal{I}^{\times 2}=\{\,I_i\cup I_j : I_i, I_j \in \mathcal{I}\,\}\), representing all pairwise unions of disjoint intervals. 
Because \(|\mathcal{I}| \le n^2\), it follows that \(|\mathcal{I}^{\times 2}| \le n^4\). 
Applying the \emph{Deterministic Separation Lemma}, we can assign to each element \(x_i\) a weight chosen \emph{deterministically} from the range \([1,n^{10}(=\gamma^2||\mathcal{I}^{\times 2}|^2)]\) so that the total weight of every set in \(\mathcal{I}^{\times 2}\) is unique. 
This replaces the randomized weight assignment used in~\cite{DBLP:conf/mfcs/BandopadhyayBMS24} with a fully deterministic construction, allowing their reduction to go through without any probabilistic assumption and thereby yielding a strengthened, completely deterministic hardness result.


\begin{theorem}Unless {\sf ETH} fails, there is no algorithm for {\ALP} running in $f({\sf dtp}(G) + |A|) \cdot n^{o({\sf dtp}(G) + |A|)}$-time.
\end{theorem}

We bring this to the attention of the reader that choosing weights from a smaller range $[1, n^8]$ already suffices. This follows because, instead of using the general bound $\gamma^2|\mathcal{F}|^2$, we can rely on the fact that $|\mathcal{I}^{\times 2}| < n^4$, where $\mathcal{I}^{\times 2}$ denotes the family of pairwise unions of disjoint intervals. It is immediate that the total number of projected sets of $\mathcal{I}^{\times 2}$, denoted by $\mathcal{I}^{\times 2}_{\leq n}$, is equal to $|\mathcal{I}^{\times 2}|$ itself. Consequently, the upper bound on the weight range can be improved to $n^8$.

\subsection{NP-hardness of {\sc Bin Filling} Problem.}

\noindent
\begin{definition}[{\sc Bin Filling} Problem]
Given a multiset of items with positive integral weights (encoded in unary), a bin capacity 
$M \in \mathbb{Z}_{>0}$, and an integer $n$, decide whether it is possible to select a submultiset 
of the items and partition it into exactly $n$ bins such that the total weight of each bin is 
exactly $M$. Items not placed into any bin may be discarded.
\end{definition}

\medskip
We establish strong NP-hardness of this problem via a reduction from the classical NP-complete problem
\emph{3-Dimensional Perfect Matching (3DPM)}, as provided by Karp~\cite{karp1972reducibility}.
 In 3DPM, the input consists of three disjoint sets 
$X,Y,Z$, each of size $n$, and a collection of triples 
$T \subseteq X \times Y \times Z$. The task is to decide whether there exists a perfect matching 
$M \subseteq T$ of size $n$ that covers every element of $X \cup Y \cup Z$ exactly once.

\medskip

\noindent\textbf{Construction of the Reduction to the {\sc Bin Filling} Problem.}  
Let $\mathcal{F}$ denote the family of all subsets of $X \cup Y \cup Z$ of size at most $3$.  
Clearly, $|\mathcal{F}| \le (3n)^3 + (3n)^2 + 3n.$
By the Deterministic \emph{Separation Lemma}, we can assign to each element of 
$X \cup Y \cup Z$ a distinct positive integral weight from the range 
$[1,n^2\cdot \big((3n)^3 + (3n)^2 + 3n\big)^2]$
such that every set in $\mathcal{F}$ has a unique total weight. For each triple $t \in T$, we introduce a \emph{large item} $e_t$ of weight $w(e_t) := M - w(t)$,
where we set $M := n^{10}$. We also introduce $n$ bins, each of capacity $M$. The original elements from $X \cup Y \cup Z$ are called \emph{small items}, 
while the newly introduced items $e_t$s are referred to as \emph{large items}.


\begin{lemma}
Every feasible packing that fills all $n$ bins contains exactly one large item in each bin.
\end{lemma}

\begin{proof}
Suppose there exists a bin that does not contain any large item. Then it can contain only small items. 
The total weight of all small items together is at most $
3n \times n^2\cdot \big( (3n)^3 + (3n)^2 + 3n \big)^2$,
which is asymptotically $\mathcal{O}(n^9)$. 
Since the bin capacity is $M = n^{10}$, any single bin without a large item cannot be completely filled. Next, consider the possibility that a bin contains two or more large items. 
Each large item $e_t$ has weight
$w(e_t) = M - \mathrm{Weight}(t) \;\;\ge\;\; n^{10} - 3\times n^2\big( (3n)^3 + (3n)^2 + 3n \big)^2>.9n^{10}>.9\times M.$
Therefore, the combined weight of any two large items is greater than $M$ exceeding the bin capacity, a contradiction.
\end{proof}



\begin{lemma}
Every bin in a feasible packing contains exactly three small items.
\end{lemma}

\begin{proof}
Suppose some bin contained at least four small items. Since there are $3n$ small items and $n$ bins, by the pigeonhole principle some other bin must then contain at most two small items. Consider such a bin $B$. It contains exactly one large item $e_t$ of weight $M-w(t)$ together with at most two small items of combined weight $w(t)$. This would imply that the weight of a set of size three equals the weight of a set of size at most two. However, both sets belong to the family $\mathcal{F}$, and by the deterministic \emph{Separation Lemma}, distinct sets have distinct weights. This yields a contradiction.
\end{proof}



\begin{corollary}
Each bin contains exactly one large item and exactly three small items.
\end{corollary}

\begin{lemma}
There exists a perfect 3D matching if and only if the constructed Bin Filling instance is feasible.
\end{lemma}

\begin{proof}
$(\Rightarrow)$ Suppose there exists a perfect 3D matching $M^* \subseteq T$. For each triple $t=\{x,y,z\} \in M^*$, place the three small items $w(x),w(y),w(z)$ together with the large item $e_t$ into a bin. Their total weight is $w(x)+w(y)+w(z)+w(e_t) = \mathrm{Weight}(t) + (M-\mathrm{Weight}(t)) = M$. Since $M^*$ has $n$ triples covering all $3n$ elements exactly once, all bins are filled (to capacity $M$) exactly.

$(\Leftarrow)$ Suppose the Bin Filling instance is feasible. By the previous lemmas, each bin contains exactly one large item and exactly three small items. Say, one bin contains a large item $e_t$. To reach capacity $M$, the three small items it contains must have weight $M-w(e_t)$. And there is precisely on triple (the triple being $t$) in $\mathcal{F}$ whose weight is $M-w(e_t)$ due to Deterministic \emph{Separation Lemma}. Thus each bin identifies a unique triple from $T$, and together these triples cover all $3n$ elements exactly once, which gives a perfect matching.
\end{proof}

\begin{theorem}
The {\sc Bin Filling} problem is strongly {\sf NP-hard}.
\end{theorem}

We remark that the same reduction strategy could also have been applied to design the hardness of the {\sc Triangle Partitioning} problem, further underscoring the versatility of the \emph{Deterministic Separation Lemma}.

\subsection{Unique Length Paths in Trees.}

We show another application of \emph{Deterministic Separation Lemma} which can guarantee a polynomial bounded assignment of edge weights to any tree, such that the total weight of every simple path is distinct.

\begin{theorem}\label{thm:tree-path-unique}
Let $T=(V,E)$ be a tree on $n$ vertices. There is a deterministic assignment of positive integer
weights $w:E\to\{1,2,\dots,\,O(n^4)\}$ such that for any two distinct simple paths $P_1,P_2$ in $T$,
\[
w(P_1)=\sum_{e\in P_1} w(e)\neq \sum_{e\in P_2} w(e)=w(P_2).
\]Moreover, the weights can be constructed in polynomial time.
\end{theorem}

\begin{proof}
Every simple path in a tree is uniquely determined by its two endpoints. Thus the number of 
unordered simple paths in $T$ equals the number of unordered pairs of distinct vertices, namely $
|\mathcal{P}| = \binom{n}{2} = \frac{n(n-1)}{2} \leq n^2,
$
where $\mathcal{P}$ denotes the family of edge-sets of all simple paths in $T$.
Apply the Deterministic \emph{Separation Lemma} (Lemma~\ref{lem:deterministic-separation}) to the set system
$(E,\mathcal{P})$. The lemma guarantees a deterministic assignment of positive integers
$$w:E\to\{1,2,\dots,n^2|\mathcal{P}|^2\}$$
such that the sum of weights of any two distinct sets in $\mathcal{P}$ are different. Since
$|\mathcal{P}| = O(n^2)$, we have $n^2|\mathcal{P}|^2 = O(n^6)$, giving the claimed polynomial bound. Hence we have a resulting edge-weight assignment satisfying the property that every simple path
in $T$ has a unique total weight.
\end{proof}

\section{Conclusion}
We presented a deterministic procedure for assigning integral values to elements of a universe so that all sets drawn from a prescribed family are uniquely distinguishable. Our construction derandomizes the earlier probabilistic version of the \emph{Separation Lemma}, which ensures that each set in the family can be assigned a unique value. Moreover, we can introduce additional elements to guarantee that there is a unique way of constructing sets to achieve a target value, any such set must precisely contain a set from the original family, using exactly all its elements together with the new elements. This idea can be useful in various hardness reductions. Beyond these applications, the approach can serve as a heuristic for estimating or counting solutions in certain combinatorial settings. Future directions include extending the framework to richer combinatorial families, thereby broadening its impact in the algorithmic and combinatorial domains.

\bibliographystyle{plain}
\bibliography{arxivversion}
\end{document}